\documentclass[twocolumn]{article}
\usepackage[margin=2cm]{geometry}
\usepackage{authblk}
\usepackage{microtype}
\usepackage{amsmath, amsthm, amsfonts}
\usepackage{xspace}
\usepackage{xcolor}
\usepackage{url}
\usepackage{forest}
\usepackage{tikz}
\usetikzlibrary{matrix,positioning,shapes.callouts,tikzmark, decorations.pathreplacing,arrows,calc,chains,patterns,patterns.meta}

\newcommand{\devnull}[1]{}

\newcommand{\scheduling}{\ensuremath{P|outtree,\allowbreak p_j=1|\sum wC}}

\newcommand{\defn}[1]{\textbf{\emph{#1}}}

\newcommand{\calP}{\mathcal{P}}
\newcommand{\calT}{\mathcal{T}}
\newcommand{\calC}{\mathcal{C}}
\newcommand{\calH}{\mathcal{H}}

\renewcommand{\epsilon}{\varepsilon}

\newcommand{\cost}{cost}
\newcommand{\costf}{cost^f}

\newcommand{\wods}{WODs\xspace}
\newcommand{\wod}{WOD\xspace}

\newcommand{\lsms}{LSM-trees\xspace}

\newcommand{\bet}{$B^{\varepsilon}$-tree\xspace}
\newcommand{\bets}{$B^{\varepsilon}$-trees\xspace}

\newcommand{\ios}{IOs\xspace}

\newcommand{\proc}[1]		{\ifmmode\mbox{\textsc{#1}}\else\textsc{#1}\fi}

\newtheorem{theorem}{Theorem}

\newtheorem{lemma}[theorem]{Lemma}

\newtheorem{observation}[theorem]{Observation}

\definecolor{niceblue}{rgb}{.392,.584,.929}
\definecolor{nicered}{rgb}{1,.375,.375}
\definecolor{nicegreen}{rgb}{.1,.8,.4}
\definecolor{nicepurple}{rgb}{.78,.26,.93}
\definecolor{niceorange}{rgb}{1,.575,.25}
\definecolor{niceyellow}{rgb}{.98,.94,.3}

\newcommand{\conferencename}[3]{
\ifx\longconferencenames\undefined
\newcommand{#1}[0]{{#2}}
\else
\newcommand{#1}[0]{{#3}}
\fi
}

\conferencename{\podc}{PODC}{Proceedings of the ACM symposium on Principles of
distributed computing (PODC)}
\conferencename{\asplos}{ASPLOS}{{Proceedings of the ACM International Conference on Architectural Support for Programming Languages and Operating Systems (ASPLOS)}}
\conferencename{\mascots}{MASCOTS}{{Proceedings of the IEEE International Symposium on Modeling, Analysis and Simulation of Computer and Telecommunications Systems (MASCOTS)}}
\conferencename{\spaa}{SPAA}{{Proceedings of the ACM symposium on Parallelism in algorithms and architectures (SPAA)}}
\conferencename{\osdi}{OSDI}{{Proceedings of the USENIX Symposium on Operating Systems Design and Implementation (OSDI)}}
\conferencename{\disc}{DISC}{{Proceedings of the International Conference on Distributed Computing (DISC)}}
\conferencename{\usenixatc}{USENIX ATC}{{Proceedings of the USENIX Annual Technical Conference}}
\conferencename{\usenixsec}{USENIX Security}{{Proceedings of the USENIX Security Symposium}}
\conferencename{\pldi}{PLDI}{{Proceedings of the ACM SIGPLAN conference on Programming language design and implementation (PLDI)}}
\conferencename{\computer}{Computer}{{IEEE Computer}}
\conferencename{\sosp}{SOSP}{{Proceedings of the ACM SIGOPS Symposium on Operating Systems Principles (SOSP)}}
\conferencename{\isca}{ISCA}{{Proceedings of the ACM IEEE International Symposium on Computer Architecture (ISCA)}}
\conferencename{\csaw}{CSAW}{{Proceedings of the ACM Workshop on Computer Security Architecture (CSAW)}}
\conferencename{\wddd}{WDDD}{{Proceedings of the Workshop on Duplicating, Deconstructing, and Debunking (WDDD)}}
\conferencename{\vldb}{VLDB}{{Proceedings of the International Conference on Very Large Databases (VLDB)}}
\conferencename{\toplas}{TOPLAS}{{ACM Transactions on Programming Languages and Systems (TOPLAS)}}
\conferencename{\tocs}{ACM TOCS}{{ACM Transactions on Computer Systems (TOCS)}}
\conferencename{\tos}{ACM TOS}{{ACM Transactions on Storage (TOS)}}
\conferencename{\ppopp}{{PPoPP}}{{Proceedings of the ACM SIGPLAN Symposium on Principles and Practice of Parallel Programming (PPoPP)}}
\conferencename{\jpdc}{J. Parallel Distrib. Comput.}{{Journal of Parallel and Distributed Computing}}
\conferencename{\ismm}{ISMM}{{Proceedings of the ACM International Symposium on Memory Management (ISMM)}}
\conferencename{\cacm}{CACM}{{Communications of the ACM (CACM)}}
\conferencename{\hpca}{HPCA}{{Proceedings of the IEEE International Symposium on High-Performance Computer Architecture (HPCA)}}
\conferencename{\transact}{TRANSACT}{{Proceedings of the ACM SIGPLAN Workshop on Transactional Computing (TRANSACT)}}
\conferencename{\iiswc}{IISWC}{{Proceedings of the IEEE International Symposium on Workload Characterization (IISWC)}}
\conferencename{\tpds}{IEEE Trans, Parallel Distrib. Syst.}{{IEEE Transactions on Parallel and Distributed Systems}}
\conferencename{\osr}{OSR}{{ACM Operating Systems Review}}
\conferencename{\nsdi}{NSDI}{{Proceedings of the USENIX Symposium on Networked Systems Design and Implementation (NSDI)}}
\conferencename{\cc}{CC}{{Proceedings of the International Conference on Compiler Construction (CC)}}
\conferencename{\surveys}{ACM Comput. Surv.}{{ACM Computing Surveys}}
\conferencename{\icde}{ICDE}{{Proceedings of the IEEE International Conference on Data Engineering (ICDE)}}
\conferencename{\fast}{FAST}{{Proceedings of the USENIX Conference on File and Storage Technologies (FAST)}}
\conferencename{\eurosys}{{E}uro{S}ys}{{Proceedings of the ACM European Conference on Computer Systems ({E}uro{S}ys)}}
\conferencename{\hotos}{HotOS}{{Proceedings of the USENIX Workshop on Hot Topics in Operating Systems (HotOS)}}
\conferencename{\hotcloud}{HotCloud}{{Proceedings of the USENIX Workshop on Hot Topics in Cloud Computing (HotCloud)}}
\conferencename{\oopsla}{OOPSLA}{{Proceedings of the ACM SIGPLAN Conference on Object-Oriented Programming, Systems, Languages, and Applications (OOPSLA)}}
\conferencename{\ndss}{NDSS}{{Proceedings of the Network and Distributed System Security Symposium (NDSS)}}
\conferencename{\oakland}{IEEE S\&P}{{Proceedings of the IEEE Symposium on Security and Privacy (Oakland)}}
\conferencename{\ispass}{ISPASS}{Proceedings of the IEEE International Symposium on Performance Analysis of Systems and Software (ISPASS)}
\conferencename{\europar}{{E}uro{P}ar}{{Proceedings of the European Conference on Parallel Programming ({E}uro{P}ar)}}
\conferencename{\sigcse}{{SIGCSE}}{{Proceedings of the ACM SIGCSE technical symposium on Computer science education (SIGCSE)}}
\conferencename{\ccs}{{CCS}}{{Proceedings of the ACM Conference on Computer and Communications Security (CCS)}}
\conferencename{\veeconf}{{VEE}}{{Proceedings of the International Conference on Virtual Execution Environments (VEE)}}
\conferencename{\lisa}{{LISA}}{{Proceedings of the Large Installation System Administration Conference (LISA)}}
\conferencename{\scool}{SCOOL}{{Proceedings of the Workshop on Synchronization and Concurrency in Object-Oriented Languages (SCOOL)}}
\conferencename{\cgo}{CGO}{{Proceedings of the International Symposium on Code Generation and Optimization (CGO)}}
\conferencename{\dsn}{{DSN}}{Proceedings of the International Conference on Dependable Systems and Networks (DSN)}
\conferencename{\sac}{{SAC}}{{Proceedings of the ACM Symposium on Applied Computing (SAC)}}
\conferencename{\cluster}{{IEEE Cluster}}{{IEEE International Conference on Cluster Computing}}
\conferencename{\soda}{{SODA}}{{ACM-SIAM Symposium on Discrete Algorithms (SODA)}}
\conferencename{\hotstorage}{{HotStorage}}{{Proceedings of the USENIX Conference on Hot Topics in Storage and File Systems (HotStorage)}}
\conferencename{\hotdep}{{HotDep}}{{Proceedings of the USENIX Conference on Hot Topics in System Dependability (HotDep)}}
\conferencename{\cikm}{{CIKM}}{{Proceedings of the ACM International Conference on Information and Knowledge Management (CIKM)}}
\conferencename{\sigmod}{{SIGMOD}}{{Proceedings of the ACM SIGMOD International Conference on Management of Data (SIGMOD)}}
\conferencename{\sigmetrics}{{SIGMETRICS}}{{Proceedings of the ACM SIGMETRICS International Conference on Measurement and Modeling of Computer Systems (SIGMETRICS)}}
\conferencename{\supercomputing}{{SC}}{{Proceedings of the ACM/IEEE Conference on Supercomputing (SC)}}

\begin{document}
\title{Root-to-Leaf Scheduling in Write-Optimized Trees}

\author[1]{Christopher Chung}
\author[2]{William Jannen}
\author[2]{Samuel McCauley}
\author[3]{Bertrand Simon}
\affil[1]{Independent Researcher, \texttt{cchung3020@gmail.com}}
	\affil[2]{Williams College, \texttt{\{jannen,sam\}@cs.williams.edu}}
	\affil[3]{IN2P3 Computing Center/CNRS, \texttt{bertrand.simon@cc.in2p3.fr}}

\date{}



  

\maketitle

\begin{abstract}
	In a large, parallel dictionary, performance is dominated by the cache efficiency of its database operations.
	Cache efficiency is analyzed algorithmically in
the disk access machine (DAM) model, in which any $P$ disjoint sets of $B$ contiguous elements can be moved simultaneously for unit cost.

  Write-optimized dictionaries (\wods) are a class
  of cache-efficient data structures
  that buffer updates and apply them in batches to optimize the amortized update cost in the DAM model.
  As a concrete example of a \wod, consider a \bet.
  \bets encode updates as messages that are
  inserted at the root of the tree.

  Messages are applied lazily at leaves:
  \bets only move---or ``flush''---messages when close to $B$ of them can be written simultaneously, optimizing the amount of work done per DAM model cost.
  Thus, recently-inserted messages reside at or near the root
  and are only flushed down the tree after a sufficient number of new messages arrive.  Other \wods{} operate similarly.

  Although this lazy approach works well for many operations,
  some types of updates do not complete until the update message reaches a leaf.
For example, both deferred queries and secure deletes must flush through all nodes along their root-to-leaf path before the operation takes effect.

What happens when 
we want to service a large number of (say) secure deletes as quickly as possible?
 Classic techniques leave us with an unsavory choice. 
 On the one hand, we can group the delete messages using a write-optimized approach and move them down the tree lazily. But, then many individual deletes may be left incomplete for an extended period of time, as their messages sit high in the tree waiting to be grouped with a sufficiently large number of related messages.
 On the other hand, we can ignore cache efficiency and perform a root-to-leaf flush for each delete. This begins work on individual deletes immediately, but performs little work per unit cost, harming system throughput.
We propose a middle ground to this all-or-nothing tradeoff.

This paper investigates a new framework for efficiently flushing collections of messages 
from the root 
to their leaves 
in a write-optimized data structure.
 We model each message as completing when it has been flushed to its leaf;  our goal is to minimize the average completion time in the DAM model.  
 We give an algorithm that $O(1)$-approximates the optimal average completion time in this model (whereas achieving the optimal solution is NP-hard).
 Along the way, we give a new $O(1)$-approximation algorithm for a classic scheduling problem: scheduling parallel tasks for weighted completion time with tree precedence constraints.
\end{abstract}

\section{Introduction}
\label{sec:introduction}

In large, parallel dictionaries,
performance hinges on the cache efficiency of database operations.
Therefore, data structures such as $B$-trees that achieve strong guarantees on their cache efficiency have become ubiquitous in modern databases.

The Disk Access Machine (DAM) model~\cite{dam} is a powerful tool for theoretically describing the cache efficiency of a data structure. In the DAM model there are three machine-dependent parameters: $B$, the size of a cache line; $P$, the number of disk accesses that can be performed in parallel; and $M \gg PB$, the size of the cache.
 In a single \defn{IO}, the data structure can read or write up to $P$ sets of $B$ contiguous elements to or from cache.  The cost of any process under the DAM model is defined as the number of \ios required to complete the process.  

Recently, write-optimized data structures have revolutionized the theory of cache-efficient data structures~\cite{BrodalFagerberg03,lsm,dayan2018optimal,bender19spaa,pandey2020timely,bender2017write,BenderDaFa20}.
A \defn{write-optimized} data structure handles inserts lazily: rather than traversing the data structure to put newly-inserted items directly in place, up to $B$ items are accumulated and moved through the data structure in large \ios.  Write-optimized data structures greatly improve insert performance in the DAM model.  For example, if the cache size $B$ is much larger than the height of the data structure, a write-optimized data structure requires $o(1)$ \ios per insertion. This improvement in write cost comes at no asymptotic cost to query time.

This performance improvement has also been realized in practice.
\lsms, \bets, and their variants
are popular \wods
that have been
used to implement file systems~\cite{jannen15fast,jannen15tos,yuan16fast,jiao22eurosys,shetty13fast,ren13atc,esmet12hotstorage},
key-value stores~\cite{LevelDB,RocksDB,lu16fast,raju17sosp},
and other applications where cache efficiency is critical.

\paragraph{\bets}

In this paper, our model and algorithmic results are specifically tailored to flushes in the \bet, a classic \wod originally presented in~\cite{BrodalFagerberg03}.
 Nonetheless, it is likely that similar strategies to those presented here would apply to other WODs, such as LSM-trees.
For example, see the discussion in~\cite{BenderDaFa20} of the similarities between LSM compaction strategies and \bet flushing policies.

We now briefly describe the structure of \bets as well as \bet insert and query operations. See e.g.~\cite{BrodalFagerberg03,bender15login} for more detailed descriptions.

\bets, like $B$-trees, are trees of nodes of size $B$.
All \bet leaves are at the same height, all \bet keys are stored in binary search tree order, and every \bet node has $\Theta(B^{\epsilon})$ children.
\bet nodes maintain internal state so that these invariants can be maintained efficiently.

Each \bet node has a \defn{buffer}
in which messages can be stored.  (Thus, a node can be read or written in a single IO.)
We ignore the internal state of a node: we assume the tree is static and that we always know the leaf where any key should be stored.  
Thus we treat each node as consisting solely of a buffer of size exactly $B$.

To insert an item into a \bet, place the item in the root node's buffer.  Then, recursively flush as follows: if the buffer of a node $v$ is full, go through all $B$ items in the buffer and determine which child they would be flushed to next.  The child with the most such items is flushed to, moving every possible item to that child; recurse if the child now has $B$ or more items.

To query an item, traverse the root-to-leaf path $p$ of the tree using the binary search tree ordering.  The binary search tree ordering guarantees that the queried item is stored along this path.  For each node $v$ in $p$, determine if any of the $B$ items in the buffer of $v$ are the queried item.

\paragraph{Parallelism}
Modern large-scale systems can often access several disjoint cache lines in parallel.  This directly models the behavior of SSDs~\cite{chen16tos,bender21topc}.  These parallel accesses have been implemented for flushes in practical WODs-based databases such as RocksDB~\cite{RocksDBcompaction}.

Following the DAM model~\cite{dam}, $P$ denotes the number of accesses that can be performed in parallel.  
Thus, up to $P$ disjoint flushes can be performed in a single IO\@.
Generally, $P$ is a small constant on real-world systems; however, we do not assume that $P$ is a constant in our results (i.e.\ our $O(1)$-approximation holds for any $P$).

\paragraph{Upserts and Write-Optimization}
Crucially, WOD performance features read-write asymmetry: write-optimization drastically speeds insertions, whereas queries cannot be write-optimized and are therefore much more expensive.  This asymmetry follows from a distinction in what each operation requires.  Insertions are flexible: the only requirement for newly inserted items is that they are findable by subsequent queries.  Queries, on the other hand, must immediately and correctly return.  We therefore cannot amortize standard query costs.

This observation leads to the question: what operations can we write optimize, i.e., encode the operations as data that is lazily flushed down the tree?  
Such operations are called \defn{upserts}.
A classic example is deletes: rather than finding and removing the element, we can insert a ``tombstone'' message that is slowly flushed down the tree, only truly removing the item when the tombstone meets the actual element~\cite{bender15login}.  This approach does not affect query correctness so long as queries search for tombstone elements---a query that finds a tombstone acts as if the element was already deleted.  With this improvement, deletes can enjoy write-optimized-insert performance.  Recent work has given similar results for a class of queries that can be deferred~\cite{derange}; by encoding a query as a deferred query---or ``derange query''---message, a WOD can delay answering the query until the time that the derange query message meets the data of interest.  

\paragraph{Flushing a Root-to-Leaf Path}
This paper addresses a broad  issue that arises when using upserts to achieve write-optimized performance: \emph{some \wod{} operations may remain incomplete until an entire root-to-leaf path has been flushed}.  We call these \defn{root-to-leaf operations}.

As already discussed, deferred queries fall under this category: we cannot answer a deferred query for an item until we have checked every buffer along its path.
However, subtler issues may also arise.
Consider deletion.
As explained above,
\bet{} items are deleted by inserting a tombstone to 
signal that the item is no longer valid.
However, many applications demand a \defn{secure delete},
where even an adversary who is given access to the system is unable to recover an item~\cite{reardon13SSP}.
Tombstone messages do not realize secure deletion. To see why, consider an adversary that is able to examine the data structure's contents offline; although a tombstone renders an item \emph{logically} invalid, the item remains \emph{physically} present until it is removed from the tree.
To realize secure deletion,
a \bet{} must therefore flush a tombstone along an entire root-to-leaf path,
purging the physical data from the leaf.

\paragraph{A New Kind of Latency}
Root-to-leaf operations motivate a new latency consideration for \wods.  Traditionally, the latency for \bet operations fall into two categories: read operations (e.g., queries) must be answered immediately via multiple \ios, whereas write operations (e.g., inserts, upserts) benefit from write-optimized batching and the  flexibility to delay flushes indefinitely.  

This leaves open the question: what happens when there is a backlog of root-to-leaf operations?  For example, a large number of deferred queries may have nearby deadlines.  Or, a firm may want to perform nightly purges of sensitive information by querying for outdated data and then secure deleting it---thus producing a large batch of root-to-leaf operations at the end of each working day.

When a large number of root-to-leaf operations arrive simultaneously, we want to complete the operations as quickly as possible.  These circumstances provide a new (and, we will see, structurally rich) set of considerations for how to flush items in a \bet.  There is some notion of write-optimization: since many operations are being completed simultaneously, we can group similar operations together to improve the amount of work done per IO\@.  But there is also a notion of latency: we must complete operations as quickly as possible, which may entail performing less-efficient flushes lower in the tree before performing flushes higher in the tree.  Known flushing techniques do not suffice: handling operations one by one (as traditionally done for queries) leads to pessimal throughput, whereas greedily grouping operations (as traditionally done for inserts/upserts) leads to terrible latency.

This poses a scheduling problem.  We can view the running of the data structure as proceeding in a sequence of \defn{time steps}, with 1 IO---and therefore $P$ flushes---per time step.  The goal is to schedule flushes to complete all operations as quickly as possible.

We emphasize that in many applications, it is crucial to minimize the average (as opposed to maximum) completion time.  For example, this maximizes the throughput of the \bet if it is experiencing a temporary backlog.  Alternatively, it gives the strongest security guarantee possible for secure deletes: as many items are deleted as possible if the data structure is compromised before all deletes could be completed.

In this paper we give an algorithm to complete an offline set of root-to-leaf operations as quickly as possible.  We model these operations as \defn{messages}: each message must be flushed to some leaf in the tree, and the message does not complete until it reaches that leaf.  The goal is to minimize the average completion time (since the number of messages is fixed, we equivalently minimize the total completion time).

\subsection{Results}
\label{sec:results}

We give an $O(1)$-approximation algorithm for the problem of scheduling write optimized tree flushes so that messages reach the leaf as quickly as possible on average.  We show that this problem is NP-hard, so a constant approximation factor is the best we can hope for in the worst case.

Our solution works by reducing this problem to a classic scheduling problem:\footnote{We formally define this scheduling problem below.  In short, the goal is to schedule tasks on $P$ machines to minimize total weighted completion time. All tasks can be completed in one time step.  The tasks have precedence constraints, where the graph of precedence constraints must be an out-directed forest.} $\scheduling$.  We give a $4$-approximation algorithm for this scheduling problem.  This is not as good as the best known approximation factor of approximately $2.41$ by Li~\cite{li2020scheduling}, but our algorithm is simpler.

\subsection{Related Work}

As mentioned above, previous work on write-optimizing non-insert operation has focused on upsert operations like deferred range queries~\cite{derange} and deletes~\cite{bender15login}.  

A wider variety of work has focused on improved flushing methods for \bets and compaction strategies for LSMs.  Indeed, much past work has benefited practical implementations (i.e.~\cite{RocksDBcompaction}) in addition to published research~\cite{dostoevsky,raju17sosp,geardb,liu23hpdc}.

A long line of work has focused on engineering traditional (not write-optimized) external-memory trees to improve performance on access sequences that are fully or partially known in advance; for example~\cite{lustiber2017tree,bent1985biased,gila2023zip,achakeev2013efficient,silberstein2008efficient,kraska2018case,faloutsos1992b,cao2023learningaugmented}.   This line of work generally focuses on changing the shape of the tree and/or updating pivots rather than efficiently moving elements down a static tree.

A long line of scheduling work has examined problems similar to $\scheduling$.  
A classic result of Horn gives an optimal algorithm for $1|outtree|\sum wC$~\cite{horn1972single}.  Meanwhile, $P|outtree|\sum wC$ is known to be NP-hard~\cite{lenstra1980complexity}.
A line of work on $P|prec,p_j = 1|\sum wC$~\cite{2approxschedule,munier1998approximation} lead to the best-known approximation ratio of $1 + \sqrt{2}$ by Li~\cite{li2020scheduling}; this remains the best known approximation ratio for $\scheduling$.
Other work has focused on practical heuristics for parallel scheduling with precedence constraints; see e.g.~\cite{kwok1999static} for a survey.

Recently, Jager and Warode gave a simple algorithm for $P|prec,\allowbreak pmtn|\sum wC$ independently of this work~\cite{JagerWarode24}.  Their algorithm is more general (it works for arbitrary rather than outtree precedence constraints, and for preemption rather than assuming $p_j = 1$) and achieves a better approximation ratio.  However, while this result is simpler and faster than the state-of-the-art LP-rounding-based method in~\cite{li2020scheduling}, it uses parametric flows, resulting in $O(n^4)$ running time compared to our simple greedy $O(n\log n)$ result.  We believe both results have independent merit.

\section{Preliminaries}
\label{sec:preliminaries}

\subsection{Problem Definition}
\label{sec:problem_definition}

An instance of the problem consists of a (static) tree $T$, a set of messages $M$, and two DAM model parameters $P$ and $B$.\footnote{The classic DAM model has an additional parameter called $M$ denoting the size of cache.  We ignore this parameter as it does not affect our results; we only use $M$ to denote messages.}
We give our running times in terms of $n = |M| + |T|$.

For each node $v\in T$, we denote the number of edges on the path between $v$ and the root as $h(v)$; we call this the \defn{height} of $v$.  
We denote the height of $T$ as $h = \max h(v)$.
We say that a node $v'$ is a descendant of $v$ if $v' = v$ or $v'$ is the child of a descendant of $v$.  An edge $e = (v_1, v_2)$ is a descendant of $v$ if $v_1$ and $v_2$ are descendants of $v$.

Each message $m$ has a \defn{target leaf}  representing the final leaf where $m$ must be sent.\footnote{We assume that all messages are being sent to leaves for simplicity; our techniques likely extend to handle messages with internal targets.}  
We assume that all leaves in the tree are at height $h$.\footnote{Our results generalize immediately so long as the average height of a target leaf is $\Omega(h)$.}

We assume that all parameters remain unchanged through the execution of the algorithm: in particular, the structure of $T$ is not changed due to rebalances, and the target leaf of each message remains the same.  Since the structure of a \bet{} only changes on inserts, we can alternatively assume that there are no inserts while messages are being flushed (or that rebalances caused by inserts are delayed until after flushes are completed).  We leave to future work to find an algorithm that can handle inserts interleaved with message flushes.

\paragraph{Flushes}
A schedule $S$ consists of a sequence of flushes, with at most $P$ simultaneous flushes per \defn{time step}.

A \defn{flush} $f$ is defined by an edge $(v, v')\in T$ where $v'$ is a child of $v$, and a set of at most $B$ messages $M_f$.\footnote{Each flush has at most $B$ messages; the parameter $B$ additionally constrains the number of messages stored in each node below.}  We also use flush as a verb: the schedule flushes the messages in $M_f$ from $v$ to $v'$.
No message may be in two different flushes in the same time step.
Thus, a flush in a schedule $S$ at time step $t$ is a \defn{valid flush} if all messages in $M_f$ are in $v$ at $t$, and every other flush $f'$ at $t$ has $M_f \cap M_{f'} = \emptyset$.
All messages in $M_f$ are moved to $v'$ after the flush ends---thus, they are in $v'$ at time step $t+1$.

\paragraph{Node sizes}

The nodes in a write-optimized tree have a size limit:
any internal non-root node $v\in T$ can contain up to $B$ messages at any time. 

However, write-optimized data structures do not obey a strict size limit when there are \defn{cascades}: a large number of messages flushed down a path toward the same node.  Consider, for example, two internal nodes, $v_1$ and $v_2$, and a leaf $\ell$, where $v_1$ is the parent of $v_2$ and $v_2$ is the parent of $\ell$.  One may want to take all messages in $v_1$ with target leaf $\ell$ and flush them to $v_2$; then, flush all messages in $v_2$ with target leaf $\ell$ to $\ell$ (see Figure~\ref{fig:cascades}).  These cascades are fundamental to write optimized tree flushing; for example, the work of Bender et al.~\cite{BenderDaFa20} is dedicated entirely to limiting the size of these cascades.  

Notice, however, that  a cascade may cause $v_2$ to \emph{temporarily} overflow in this example if $v_2$ already contained close to $B$ messages.  Because the overflow is temporary, tree implementations permit cascades in practice (overflowing messages can be handled with a small amount of scratch space).

\tikzset{three sided/.style={
        draw=none,
        append after command={
            [shorten <= -0.5\pgflinewidth]
            ([shift={(-0.0\pgflinewidth,+0.3\pgflinewidth)}]\tikzlastnode.north east)
			edge[dotted,thick]([shift={( 0.0\pgflinewidth,+0.3\pgflinewidth)}]\tikzlastnode.north west) 
            (\tikzlastnode.north east)
			edge[dotted,thick](\tikzlastnode.south east)            
            ([shift={( 0.0\pgflinewidth,-0.3\pgflinewidth)}]\tikzlastnode.south west)
			edge[dotted,thick]([shift={(0.0\pgflinewidth,-0.3\pgflinewidth)}]\tikzlastnode.south east)
        }
    }
}

\definecolor{niceblue}{rgb}{.392,.584,.929}

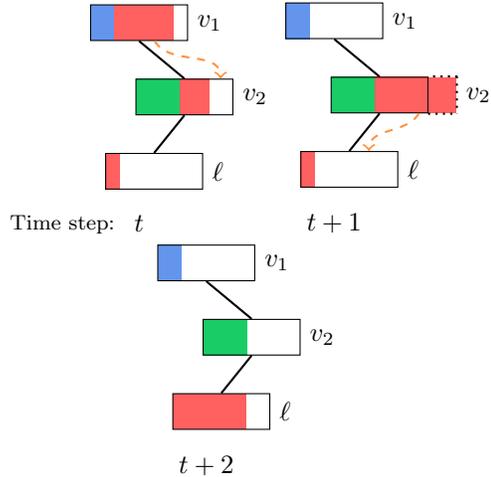
\begin{figure}[ht]
	\begin{center}
	\begin{tikzpicture}
		\node[draw,label=right:$v_1$] 
			[path picture={ 
			\fill[niceblue](root.south west)rectangle ($(root.north west)!.25!(root.north east)$); 
			\fill[nicered] ($(root.south west)!.25!(root.south east)$) rectangle ($(root.north west)!.85!(root.north east)$);
		}]
				(root){\phantom{AAAA}};
    \node[draw,below right = .5cm and -.7cm of root,label=right:$v_2$] 
			[path picture={ 
			\fill[nicegreen](child.south west)rectangle ($(child.north west)!.45!(child.north east)$); 
			\fill[nicered] ($(child.south west)!.45!(child.south east)$) rectangle ($(child.north west)!.75!(child.north east)$);
		}]
				(child){\phantom{AAAA}};
    \node[draw,below left = .5cm and -.9cm of child,label=right:$\ell$] [path picture={ \fill[nicered](leaf.south west)rectangle ($(leaf.north west)!.15!(leaf.north east)$); }](leaf){\phantom{AAAA}};
	\node[below = 2.2cm of root](t){$t$};
	\node[left=.01cm of t]{\footnotesize Time step:};
	\path[draw,thick] (root.south) -- (child.north);
	\path[draw,thick] (child.south) -- (leaf.north);
	\draw[thick,dashed,->,out=300,in=90,niceorange] 
		($(root.south west)!.66!(root.south east)$) 
		to  
		($(child.north west)!.87!(child.north east)$);
	\end{tikzpicture}
	\begin{tikzpicture}
		\node[draw,label=right:$v_1$] 
			[path picture={ 
			\fill[niceblue](root.south west)rectangle ($(root.north west)!.25!(root.north east)$); 
		}]
				(root){\phantom{AAAA}};
    \node[draw,below right = .5cm and -.7cm of root] 
			[path picture={ 
			\fill[nicegreen](child.south west)rectangle ($(child.north west)!.45!(child.north east)$); 
			\fill[nicered] ($(child.south west)!.45!(child.south east)$) rectangle (child.north east);
		}]
				(child){\phantom{AAAA}};
			\node[draw=black,fill=nicered,right=-.00cm of child,label=right:$v_2$,three sided]{\phantom{I}};
    \node[draw,below left = .5cm and -.9cm of child,label=right:$\ell$] [path picture={ \fill[nicered](leaf.south west)rectangle ($(leaf.north west)!.15!(leaf.north east)$); }](leaf){\phantom{AAAA}};
	\node[below = 2.2cm of root]{$t+1$};
	\path[draw,thick] (root.south) -- (child.north);
	\path[draw,thick] (child.south) -- (leaf.north);
	\draw[thick,dashed,->,out=240,in=90,niceorange] 
		($(child.south west)!.90!(child.south east)$) 
		to  
		($(leaf.north west)!.70!(leaf.north east)$);
	\end{tikzpicture}
	\begin{tikzpicture}
		\node[draw,label=right:$v_1$] 
			[path picture={ 
			\fill[niceblue](root.south west)rectangle ($(root.north west)!.25!(root.north east)$); 
		}]
				(root){\phantom{AAAA}};
    \node[draw,below right = .5cm and -.7cm of root,label=right:$v_2$] 
			[path picture={ 
			\fill[nicegreen](child.south west)rectangle ($(child.north west)!.45!(child.north east)$); 
		}]
				(child){\phantom{AAAA}};
    \node[draw,below left = .5cm and -.9cm of child,label=right:$\ell$] [path picture={ \fill[nicered](leaf.south west)rectangle ($(leaf.north west)!.75!(leaf.north east)$); }](leaf){\phantom{AAAA}};
	\node[below = 2.2cm of root]{$t+2$};
	\path[draw,thick] (root.south) -- (child.north);
	\path[draw,thick] (child.south) -- (leaf.north);
	\end{tikzpicture}
\end{center}
	\caption{This figure shows three successive time steps during a cascade of three nodes $v_1$, $v_2$, and $\ell$.  The flush that will occur in the next time step is shown with an orange dotted line.  Messages that have $\ell$ as a target leaf are represented in red; all others are represented in green or blue.  $v_2$ temporarily overflows on the second time step, allowing all messages to be flushed in two time steps.}
	\label{fig:cascades}
\end{figure}
To allow for schedules that include cascades, nodes may \defn{overflow}: a node may contain more than $B$ messages, so long as at most $B$ of them remain in the node at the next time step.  
Allowing overflows this way reflects the usage of write optimized trees with cascades: only $B$ messages can be written to a node, but a cascade of flushes to a common location can take place even if the cascade causes nodes to temporarily have size more than $B$.
We emphasize that the limit that only $B$ messages can participate in a flush remains a strict limit.

Thus, we formally define the \defn{space requirement} of each node as follows.  For each internal, non-root node $v\in T$ and each time step $t$, 
there can be at most $B$ messages that are in $v$ at both time steps $t$ and $t+1$.  
We allow the root and the leaves to hold an unlimited number of messages.

\paragraph{Valid Schedules and Cost}

We say that a schedule $S$ is \defn{valid} if all flushes are valid, all messages are flushed to their target leaf during $S$, and all nodes satisfy the space requirement.

A schedule $S$ is \defn{overfilling} if all flushes are valid and all messages are flushed to their target leaf during $S$.  
An overfilling schedule may have nodes that do not satisfy the space requirement.
An overfilling schedule is not a solution to our problem,
but building an overfilling schedule is a useful intermediate step in our results.

Let $c(S, m)$ denote the completion time of any message $m$ (i.e.\ the time when $m$ reaches its target leaf) in a schedule $S$.
We also say that $m$ \defn{finishes} at $c(S,m)$ in $S$.
The cost of a schedule is the total completion time: ${c(S) = \sum_{m\in M}  c(S, m)}$.

The goal of this model is to, for any $(T, M, P, B)$, find the valid schedule that minimizes the total completion time. 
Throughout the paper we call this model \defn{write-optimized root-to-leaf message scheduling}, or \defn{WORMS}.

\subsection{Scheduling with Tree Precedence Constraints}%
\label{sec:scheduling_with_tree_precedence_constraints}

Our algorithm works by reducing WORMS to a classic scheduling problem, $\scheduling$, which we define and discuss here.

An instance of $\scheduling$ consists of a set of \defn{tasks}.  Each task $t$ has processing time $1$ (i.e.\ $p_j = 1)$; therefore, a task can be completed by assigning it to one machine for one time step.  
Each task $j$ has a weight $w(j)$. 
The objective $\sum wC$ is to minimize the total weighted completion time of all tasks: if task $j$ completes at time $c(\sigma, j)$ under any schedule $\sigma$,  
the cost of a schedule $\sigma$ is $\cost(\sigma) = \sum_j c(\sigma,j) w(j)$.
The goal is to find the schedule $\sigma_O$ minimizing $\cost(\sigma_O)$.

The $P$ in $\scheduling$ represents that there are $P$ machines, and therefore up to $P$ tasks can be completed at each time step. 
(This is a slight abuse of notation as $P$ also represents the number of parallel flushes in the WORMS model.  However, $P$ has the same value in both: a WORMS instance with parameter $P$ results in a $\scheduling$ instance with $P$ machines.)

The $outtree$ in $P|outtree, p_j = 1| \sum wC$ represents that the tasks have \defn{outtree precedence constraints}.  For each task $j$, there is at most one task $j'$ such that $j$ cannot be processed until $j'$ is completed; thus $j'$ must be scheduled in a time step strictly before $j$ in any schedule.    There cannot be cycles in these constraints (as no task in a cycle can ever be completed). 
The limitations that there is at most one $j'$, and no cycles, are why these are ``outtree'' precedence constraints.  We refer to one task as a parent/child/descendant (etc.) of another in the natural way according to the precedence constraints in this tree.

\section{Using $\scheduling$ to Solve WORMS}%
\label{sec:structure}

In this section we reduce WORMS to a classic scheduling problem: $\scheduling$.  We proceed in two steps.  
First, we show how to convert any overfilling schedule to a valid schedule while losing only a constant factor in cost.
With this result in hand, we show that every overfilling schedule is reducible to to an instance of $\scheduling$.

\subsection{Converting Overfilling Schedules to Valid Schedules}
\label{sec:converting_overfilling_schedules_to_valid_schedules}

In this section we show that overfilling schedules can be converted to valid schedules while only increasing the cost by a constant factor.
\begin{lemma}
\label{lem:overfilling_to_valid}
There is a constant $c_1$ such that
for any overfilling schedule $S$ for a WORMS instance $(T,M,P,B)$, we can in $O(n\log n)$ time give a valid schedule $\hat{S}$ for WORMS satisfying
$
	\cost(\hat{S}) \leq c_1\cdot \cost(S)
	$.
\end{lemma}
The remainder of this subsection is dedicated to proving this lemma.  We fix an overfilling schedule $S$ for a WORMS instance $(T, M, P, B)$ and build a valid schedule $\hat{S}$.

\paragraph{Defining Packed Nodes}
We define the packed nodes of $(T,M,P,B)$ via a bottom-up traversal of $T$ as follows.  
For each $h'$ from $h$ to $1$, for each node $v$ at height $h'$, 
let $C(v)$ be the set of all messages $m$ such that: (1) the target leaf of $m$ is a descendant of $v$, and (2) $m$ is not in $C(v')$ for any packed node $v'$ with $h(v') > h'$.  
Then $v$ is a \defn{packed node} if $|C(v)| \geq B/6$.
We call $C(v)$ the \defn{packed contents} of $v$. 
The root is a packed node, and its packed contents are all messages that are not in another packed node.  
Therefore, each message is in the packed contents of exactly one packed node.

\paragraph{Packed Sets}

We further subdivide packed nodes based on $S$.  
For any packed node $v$ with packed contents $C(v)$, if $v$ is a leaf, 
order messages in $C(v)$ by their completion time in $S$, and then partition $C(v)$ arbitrarily into packed sets with total weight between $B/6$ and $B/2$.
If $v$ is an internal node, order the messages within $C(v)$ by the time when the messages are flushed from $v$ to one of its children in $S$.  
Then, partition the messages into sets with total weight between%
\footnote{This can be done greedily: consider each child of $v$ in turn, adding to the current packed set until it has size at least $B/6$; since each child of $v$ has at most $B/6$ remaining messages, each resulting set has size at most $B/3$.  There may be extra children at the end with at most $B/6$ total messages; add them to the final packed set, giving size at most $B/2$.} 
$B/6$ and $B/2$
such that if two messages are flushed to the same child of $v$ at the same time they are in the same packed set. 
We call these the \defn{packed sets} of $v$, and we say that $v$ is the \defn{packed parent} of each such set, as well as each message in each set.  
We will use the term ``packed sets'' 
(without referring to a specific packed node)
to refer to all packed sets of all packed nodes.

For any node $v$, let $\calC(v) = C_1, C_2, \ldots$ be the sequence of packed sets of $v$, sorted by the last time any message in the packed set is flushed out of $v$.  

\begin{figure}[ht]
	\begin{forest}
	  for tree={%
	    inner sep=0pt,
		circle,
		minimum size=2ex,
	    draw,
		s sep=1.2ex,
	  },
	  [,label=left:$3$,name=l1,color=black,draw, thick
	  [,color=nicepurple,fill
	  	[,color=black,draw
			[,color=black,draw
				[,label=below:$40$,color=black,draw,thick]
				[,label=below:$3$,color=black,draw]
			]
			[,label=right:$11$,color=black,draw,thick
				[,label=below:$5$,color=niceblue,fill ]
				[,label=below:$6$,color=niceblue,fill ]
			]
		]  
		[,label=right:$36$,color=black,draw,thick
			[,color=niceorange,fill
				[,label=below:$6$,color=black,draw ]
				[,label=below:$3$,color=black,draw ]
			]
			[,color=niceorange,fill
				[,label=below:$9$,color=black,draw ]
			]
			[,color=niceyellow,fill
				[,label=below:$9$,color=black,draw ]
			]
			[,color=niceyellow,fill
				[,label=below:$4$,color=black,draw ]
				[,label=below:$5$,color=black,draw ]
			]
		]
	  ]
	  [,label=right:$23$,color=black,draw,thick
	  	[,color=purple,fill
	    	[,color=black,draw
	    		[,label=below:$5$,color=black,draw ]
	    		[,label=below:$3$,color=black,draw ]
	    	]
	    	[,color=black,draw
	    		[,label=below:$1$,color=black,draw ]
	    	]
	    ]  
	    [,color=purple,fill
	    	[,label=right:$14$,color=black,draw,thick
	    		[,label=below:$6$,color=nicegreen,fill ]
	    		[,label=below:$8$,color=nicegreen,fill ]
	    	]
	    	[,color=black,draw
	    		[,label=below:$3$,color=black,draw]
	    		[,label=below:$3$,color=black,draw]
	    	]
	    ]
	  ]
	  ]
	  \node[right=7ex of l1,minimum width=6ex,minimum height=3ex] (label) {$B = 60$};
	\end{forest}
	\caption{Packed sets of an example WORMS instance.  Each leaf $\ell$ is labelled with the number of messages that have $\ell$ as a target leaf.  Packed parents are bolded and labelled with the size of their packed contents.  Children of an internal packed parent are filled and are colored according to the packed set their messages belong to.}
	\label{fig:packed_sets}
\end{figure}
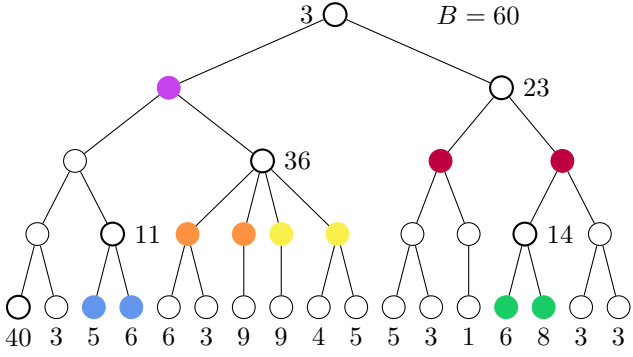

Consider a packed node $v$ with packed sets $\calC(v)$; let $C_i$ be a packed set in $\calC(v)$.
If $v$ is a leaf, define the starting time of $C_i$ to be the time step when the $B/12$th message in $C_i$ arrives at $v$.
If $v$ is an internal if $i > 1$, define the \defn{starting time} of $C_i$, denoted $\tau_i$, to be the time step when the last message in $C_{i-1}$ is flushed out of $v$ in $S$.  If $i = 1$, define the starting time to be the time step $\tau_1$ after which at least $B/12$ messages in $C_1$ have been flushed out of $v$.  

When building $\hat{S}$, the plan will be for all messages in a packed set to arrive at their packed parent at their starting time.  This may delay the message's finishing time in $\hat{S}$ past its finishing time in $S$.
Thus, we give two properties of the starting time of a packed set that allow us to charge these delays.
First, by a packed set's starting time, $\Omega(B)$ messages must have been flushed to a packed parent.  Second, the total starting times lower bound the cost of $S$.  

\begin{lemma}%
\label{lem:finishing_times_jobs_flushed}
If $C_i\in \calC(v)$ is a packed set with packed parent $v$ and start time $\tau$, at least $B/12$ messages in $C_i$ or $C_{i-1}$ were flushed to $v$ strictly before $\tau$ in $S$.
\end{lemma}
\begin{proof}
	By definition: if $v$ is internal, we define the starting time such that $B/12$ messages from $C_i$ or $C_{i-1}$ were flushed \emph{out} of $v$ by $\tau$; since $S$ is overfilling they must have been flushed to $v$ before this time.  If $v$ is a leaf, $B/12$ messages in $C_i$ were flushed to $v$ by $\tau$.
\end{proof}

\begin{lemma}%
\label{lem:finishing_times_lower_cost}
For any message $m$, let $\tau(m)$ be the starting time of the packed set of $m$.  Then $\sum_{m\in M} \tau(m) \leq 2c(S)$.

\end{lemma}
\begin{proof}
	Consider a packed set $C_i\in \calC(v)$ with starting time $\tau_i$ and packed parent $v$.  If $i > 1$ and $v$ is internal, then each message $m\in C_i$ must finish after time step $\tau_i$ by construction.  If $i = 1$, or $v$ is a leaf, then all but the first $B/12$ messages in $C_i$ must finish after $\tau_i$ by construction.  Since $|C_i| \geq B/6$, at least half the messages in $C_i$ must finish after $\tau_i$.

	In other words, $\sum_{m\in C_i} c(S,m) \geq \tau_i(|C|/2)$.
	Thus, 
\[
	\sum_{m\in M} \tau(m) = 
	\sum_{C_i} \tau_i|C_i| \leq 
	\sum_{C_i} \sum_{m\in C_i} 2c(S, m) = 2c(S).\qedhere
\]
\end{proof}

\paragraph{Three Partial Schedules}
We define $\hat{S}$ using three schedules of flushes: $U$, $U_r$ and $L$.  Intuitively, $U$ handles all flushes of messages from the root to their packed node and $U_r$ is a modification of $U$ that can be more easily merged with $L$; $L$ handles all flushes of messages from their packed node to their target leaf.

Notice that none of $U$, $U_r$, or $L$ are necessarily valid or overfilling: $U$ and $U_r$ may not flush a message to its target leaf, and $L$ may flush a message $m$ from a node $v$ without flushing $m$ to $v$ first.  
Only when we combine $U_r$ and $L$ will we create a valid schedule $\hat{S}$.

First, we define $U$.  In short, we greedily schedule flushes of each packed set to its packed parent in order by its starting time.
More formally, consider each packed set $C$ in order by its starting time $\tau$, breaking ties arbitrarily.
Let $h(C)$ be the height of $v$, the packed node containing $C$.  We would like to schedule $h(C)$ flushes, each containing the messages in $C$, from the root to $v$ beginning at $\tau - h + 1$; thus the messages in $C$ arrive in $v$ at $\tau$.  If this is possible---there are not $P$ flushes already scheduled at time $\tau - h(C) + 1$---then we add these flushes to $U$.  Otherwise, we add these flushes to $U$ at the earliest time step when there are less than $P$ already-scheduled flushes.
Creating $U$ in this way requires $O(n\log n)$ time.

Inductively, all flushes of a packed set $C$ are in $h(C)$ consecutive time steps in $U$.
In particular, because we schedule using increasing time steps, if we find an empty slot on some machine, there will be no later non-empty slots on that machine, so we can always schedule all $h$ time steps consecutively in $U$ on the same machine.

This greedy schedule guarantees that each packed set arrives at its packed node close to its starting time.

\begin{lemma}%
\label{lem:packed_set_arrival}
For any packed set ${C}$ with starting time $\tau$ and packed parent $v$, all messages in ${C}$ arrive at $v$ in $U$ on or before time step $13\tau$.
\end{lemma}
\begin{proof}

Let $\delta(C)$ be the difference between $\tau - h(C) + 1$ and the time when the first flush of $C$ begins.
	It is sufficient to prove that $\delta({C}) \leq 12\tau$.

By construction there are $P$ flushes scheduled during all time steps between $\tau - h(C) + 1$ and $\tau - h(C) + \delta({C})$.  Let $\mathcal{C}_t$ be the set of packed sets that have a flush during one of these time steps.  
	Any packed set $C$ has at most $h(C)$ flushes in $U$, and $P$ flushes are performed in that $\delta({C}) - 1$ time steps starting at $\tau$.
	Then we have that 
	\[
	\delta({C})  - 1 \leq \sum_{C \in \mathcal{C}_t} h(C)/P.
	\]

Let $\mathcal{C}_E$ consist of all packed sets with starting times before $\tau$, or equal to $\tau$ with an earlier tiebreaker.  Then $\mathcal{C}_t \subseteq\mathcal{C}_E$.

	By Lemma~\ref{lem:finishing_times_jobs_flushed}, for any $C$ in $\mathcal{C}_E$, there must have been at least $B/12$ messages flushed $h(C)$ times by $\tau$.  $S$ can flush at most $PB$ messages in a single time step; thus
	\[
	\left(\sum_{C\in \mathcal{C}_E} \frac{B}{12}\cdot h(C) \middle/PB \right) = \sum_{C\in \mathcal{C}_E} h(C)/12P \geq \tau.
	\]
	Combining, $\delta({C}) \leq 12\tau$.
\end{proof}

Now, we define $L$.
To begin, we define vocabulary based on $S$.
We call a message $m$ in a flush $f\in S$ on an edge $(v_1, v_2)$ a \defn{lower message} if $m\in C(v)$ for a packed node $v$ where $v_1$ is a descendant of $v$.  We call $f$ a \defn{lower flush} if it contains a lower message.  
By the definition of a packed set, all lower messages in a flush have the same packed parent.

To define $L$, iterate through the time steps $t$ of $S$ from $t=1$ to the end.  
Consider each of the $\leq P$ flushes $f_i$ from $v_i$ to $v_i'$ at $t$.  For each such $f_i$, if $v_i$ is the packed parent of all lower messages in $f_i$, schedule a flush in $L$ from $v_i$ to $v_i'$ 
 
in the earliest available time step\footnote{I.e.\ the first time step with fewer than $P$ flushes already scheduled.} after $27\tau_i$, consisting of all lower messages in $f_i$.
Otherwise, let $t'$ be the earliest time step in $L$ after which all lower messages in $f_i$ are in $v_i$; schedule a flush in $L$ from $v_i$ to $v_i'$ in the earliest available time step after $t'$, consisting of all lower messages in $f_i$.
$L$ can be created in $O(n\log n)$ time.

We define $L$ this way to ensure the following lemma bounding the time step when each lower flush is scheduled.

\begin{lemma}%
\label{lem:lower_flushes_time}
Any message $m$ that is a member of a packed set $C$ with starting time $\tau$ reaches its target leaf in $L$ by time step $27\tau + h + c(S, m) $.
\end{lemma}
\begin{proof}
	Consider the flushes of $m$ in $L$.  We say that $m$ is \defn{delayed} each time a flush containing $m$ cannot be scheduled in the desired time slot in $L$.  If $m$ is not delayed, it completes at time step $27\tau + (h - h(v))$ in $L$.  If all flushes containing $m$ are delayed a total of $\delta$ time steps, $m$ completes at time $27\tau + (h - h(v)) + \delta$.

	Each time $m$ is delayed, a flush $f$ (in $S$) containing $m$ cannot be scheduled due to $P$ other flushes in $L$ already being scheduled at some time step $t$.  Each of those flushes in $L$ must occur before, or at the same time as, $f$ in $S$; thus each flush must occur at or before $c(S, m)$ in $S$.  There can be at most $Pc(S,m)$ such flushes.  Thus, $\delta \leq c(S,m)$.
\end{proof}

Now, we want to combine $U$ and $L$ to obtain a valid schedule.  
To begin we observe the following: combining the flushes from both schedules gives a valid sequence of flushes for any individual message.  

\begin{observation}%
\label{obs:partial_schedules_correct}
For any message $m$, 
create a schedule $S_m$, such that if $m$ is flushed from $v_1$ to $v_2$ in $U$ or $L$, then $m$ is flushed from $v_1$ to $v_2$ in $S_m$.  Then in $S_m$, if $m$ is flushed from a node $v$ at time step $t$ then $m$ is flushed to $v$ at some time step $t' < t$; furthermore, $m$ is flushed to its target leaf $\ell_m$ by the end of $S_m$.
\end{observation}
\begin{proof}
	Let $C$ be the packed set of $m$ and $v$ be its packed parent; let $\tau$ be the starting time of $C$.  Then all messages in $C$ are flushed to $v$ by time step $13\tau$ in $U$ by Lemma~\ref{lem:packed_set_arrival}.  Each message in $C$ is flushed out of $v$ at time step $27\tau$ or later by definition.

Since $S$ is overfilling, there is a sequence of flushes containing $m$ from $v$ to its target leaf; each flush in $S$ from node $v_1$ to $v_2$ occurs at a time step when $m$ is in $v_1$.  
Since we add flushes to $L$ in order of time steps in $S$ (waiting for any messages in the flush to arrive if necessary), all flushes in $S$ are in $L$, and they are in the same order they were in $S$, and we do not flush a lower message out of a out of a node until it is flushed into the node.
\end{proof}

However, there remains an issue which motivates us to alter $U$.  As-is, when a packed set arrives at an internal packed node $v$ in $U$, it may be that there are still messages in the node that were not flushed out of $v$; thus, there may not be room for these new messages.  (This is because a flush in $L$ may be delayed far past the corresponding flush in $S$.)

To resolve this, we modify $U$ to create a schedule $U_r$ that ensures there is enough room in each packed parent for an arriving packed set.  In short, to create $U_r$, we add additional flushes to $U$ to guarantee that each node flushed to has space.

Specifically, for each packed set $C_i$ with an internal packed parent, let $\hat{\tau_i}$ be the time step that $C_i$ arrives at its packed node $v_i$.  Perform the following for all packed sets, in order by $\hat{\tau_i}$.  Immediately before $\hat{\tau_i}$, insert all flushes in $L$ from $v_i$ that are performed after time step $\hat{\tau_i} - h$.  (We subtract $h$ because, in short, we will delay flushes in $L$ by up to $h$ time steps when constructing $\hat{S}$ below.)  These flushes can all be performed in parallel; thus, if we insert $x$ flushes, all later flushes in $U$ are delayed by $\lceil x/P\rceil$ time steps.
We can create $U_r$ in $O(n\log n)$ time.

Let $U_r$ be the schedule obtained after this modification.  
By construction, every time a packed set arrives at its packed parent $v$ in $U_r$, $v$ contains no other messages from a packed set of $v$.
Furthermore, $U_r$ satisfies the following lemma.

\begin{lemma}%
\label{lem:upper_flushes_modified_correct}
For any packed set $C$ with starting time $\tau$, all messages in $C$ arrive at their packed parent $v$ in $U_r$ at time step $\hat{\tau}$ satisfying $\hat{\tau} \leq 27\tau + h$.
\end{lemma}
\begin{proof}
For any time step $t$, let $x_t$ be the number of flushes inserted immediately before $t$ into $U$ to obtain $U_r$---i.e.\ $x_t$ is the sum, over all packed parents $v_\ell$ with packed sets arriving at $t$, of the number of flushes out of $v_\ell$ in $L$ that occur after $t - h$.  By definition, if a flush occurs at time step $t'$ in $U$, it occurs at time step $t' + \sum_{t \leq t'} \lceil x_{t}/P\rceil$ in $U_r$.  Let $t_U$ be the time step when all messages in $C$ arrive at $v$ in $U$; thus we want to show that $t_U + \sum_{t \leq t_U} \lceil x_{t}/P\rceil \leq 27\tau$.
  
Let $C_i,C_{i+1}\in \calC(v')$ be two packed sets of some node $v'$; let $\tau_{i+1}$ be the starting time of $C_{i+1}$.  By definition, all flushes of messages in $C_i$ out of $v$ occur before $\tau_{i+1}$ in $S$.  Then, since the time step at which all packed sets in $U$ arrive at their packed parent is in order by starting time,
$\lceil \sum_{t \leq t_U} x_t/P\rceil \leq \tau + h$.

By Lemma~\ref{lem:packed_set_arrival}, $t_U \leq 13\tau$.  
Then we can upper bound when $C$ arrives in $U_r$.
\begin{align*}
	t_U + \sum_{t \leq t_U} \lceil x_{t}/P\rceil &\leq
	t_U + \left\lceil \sum_{t \leq t_U} x_{t}/P + 1 \right\rceil \\
												 &\leq t_U + (\tau + h +  t_U) \\
												 &\leq 27\tau + h.\qedhere
\end{align*}
\end{proof}

\paragraph{Defining $\hat{S}$}
Divide $U_r$ and $L$ into \defn{epochs} of $h$ time steps; we use epoch $i \in \{0, 1, \ldots \}$ to refer to time steps $\{hi + 1, \ldots, hi + h\}$.  Each flush in epoch $i$ of $U_r$ or $L$ will occur at some time step in $\{3hi + 1, \ldots,  3hi +4h \}$ of $\hat{S}$.

Iterate through the epochs starting at $0$.  Consider a sequence of flushes from the root to some packed node $v$ in $U_r$ where the first flush occurs at time step $hi + t$ in $U_r$ for $1 \leq t \leq h$; perform these flushes starting at time step $3hi + h + t$ in $\hat{S}$.  These flushes are all performed by time step $3hi + 2h$.  

For any flush $f$ from $v_1$ to $v_2$ in $L$ at time step $hi + t$ in $L$, 
obtain $M_f'$ by removing all messages in $M_f$ that are flushed from $v_1$ to $v_2$ in $U_r$.  Then schedule a flush in $\hat{S}$ of all messages in $M_f'$ at time step $3hi + 2h + t$.  These flushes all complete by time step $3hi + 3h$.

Creating $\hat{S}$ from $U_r$ and $L$ requires $O(n)$ time.

\paragraph{Analysis of $\hat{S}$}
From the above, $\hat{S}$ can be created (including the time to create $U_r$ and $L$) in $O(n\log n)$ total time.
We are left to prove that $\hat{S}$ is a valid schedule, and that its cost is bounded by the cost of $S$.  

\begin{proof}[Proof of Lemma~\ref{lem:overfilling_to_valid}]
	To show that $\hat{S}$ is valid, we must first show that each flush only consists of messages that are already in the correct node, and that all messages are flushed to their leaf.  

	Consider the flushes containing a message $m$.  By construction, all flushes containing $m$ in $U_r$ occur in the same order in $\hat{S}$ and all flushes containing $m$ in $L$ occur in the same order in $\hat{S}$.  Let $i$ be the epoch in which $m$ is flushed to its packed parent. 
By Observation~\ref{obs:partial_schedules_correct}, there is a flush containing $m$ for each edge from the root to its target leaf, 
all flushes containing $m$ in $U$ occur in epoch $i$ or earlier, and all flushes containing $m$ in $L$ occur in epoch $i$ or later.  
All flushes in $L$ in epoch $i$ occur later than all flushes in $U_r$ in epoch $i$.  
Therefore, all flushes contain messages already in the correct node, and all messages are flushed to their leaf.

In each of $U_r$ and $L$, there are at most $P$ flushes in any time step, and each flush contains $B$ messages.  Each time step of $\hat{S}$ contains messages from flushes in from exactly one of $U_r$ or $L$, so $\hat{S}$ also has at most $P$ flushes in each time step, and each flush contains at most $B$ messages.

Thus, $\hat{S}$ is overfilling; we must further show that it is valid.  Consider a node $v$ at time step $t$.  There can be at most $B/2$ messages with a packed parent that is an ancestor $v'$ of $v$ with $v' \neq v$, because every descendant of a packed parent contains at most $B/2$ messages from that packed parent from the definition of the packed sets.
If a message $m$ has a packed parent $v'$ that is a descendant of $v$ with $v' \neq v$, then $m$ will be flushed out of $v$ at time step $t+1$ by definition of $U$, $U_r$, and $\hat{S}$---in particular, all flushes of messages to their packed parent in $\hat{S}$ are consecutive.  (Thus, $v$ may overflow, but still satisfies the space requirement.)
Finally, if the packed parent of $m$ is $v$, let $C_i$ be the packed set of $m$.  
Let $t'$ be the time when $m$ arrives at $v$ in $U_r$.
By definition of $U_r$, any message not in $C_i$ is either (1) flushed out of $v$ in $U_r$ before $t'$, in which case it is flushed out before $t$ in $\hat{S}$, or (2) flushed out of $v$ in $L$ by time step $t' - h$; this is an earlier epoch than $t'$, so again this flush occurs before $t$.
Putting the above together,
the messages in $v$ at time step $t$ that are still in $v$ at time step $t+1$ consist of at most $B/2$ messages whose packed parent is $v$, and at most $B/2$ messages whose packed parent is an ancestor of $v$, for $B$ messages total.

Finally, we show the cost.
If a message $m$ is completed at time step $c(m, S)$ in $S$, then by Lemma~\ref{lem:lower_flushes_time} it is completed at time step $27\tau + c(m, S) + h$ in $L$ or $U_r$ ($m$ is completed in $U_r$ if its packed parent is a leaf); therefore it is completed by time step $3(27\tau + c(m, S) + h) + 3h$ in $\hat{S}$.  
For all $m$, $c(m, S) \geq h$.
Thus we can bound
\[
	c(\hat{S}) \leq \sum_{m\in M} 81\tau_m + 7c(S, m).
\]
Replacing with Lemma~\ref{lem:finishing_times_lower_cost},
\[
	c(\hat{S}) \leq 162 c(S) + \sum_{m\in M} 7c(S, m) \leq 169c(S).
\]
Setting $c_1 = 169$ we are done.
\end{proof}

\subsection{Reduction}%
\label{sec:reduction}

Now we are ready to show how to reduce WORMS---in particular, the problem of finding an overfilling schedule for WORMS---to $\scheduling$.

To begin, we describe the reduction. For any WORMS instance $(T,M,P,B)$, we describe how to create an instance of $\scheduling$ which we denote $\calT(T,M,P,B)$.  
First, we define the \defn{oblivious packed sets}---these serve a function similar to the packed sets in Section~\ref{sec:converting_overfilling_schedules_to_valid_schedules}; however, the packed sets are defined using a schedule $S$, whereas the oblivious packed sets depend only on $(T,M,P,B)$.  To create the oblivious packed sets, for each internal packed node $v\in T$, we partition the children of $v$ arbitrarily into sets such that the total number of descendants of all children in a set is between $B/6$ and $B/2$.  
(This can be done using the same method as the packed sets, but we no longer order the children of $v$ by the time step they are flushed in some schedule $S$.)  
If $v$ does not have children, we partition the messages arbitrarily into sets of size between $B/6$ and $B/2$.

For each oblivious packed set $C$ with packed parent $v\in V_P$, we create a chain of $h(v)$ tasks with weight $0$ in $\calT(T,M,P,B)$; each has the previous as a precedence constraint.  The first task in the chain has no precedence constraint.  Let\footnote{We use $j$ to denote tasks since $t$ denotes time.} $j'$ be the last task in the chain.  
 
If $v$ is an internal node,
we recursively copy the subtree rooted at $v$ in $T$ to $\calT(T,M,P,B)$.
Specifically, if $c$ is the number of children of $v$ in $T$, create $c$ tasks with weight $0$ and $v$ as a precedence constraint (one for each outgoing edge of $v$ in $T$), then do the same recursively for each such child.  
When $v$ is a leaf $\ell\in T$, we create a task $j_v$ in $\calT(T,M,P,B)$ representing the in-edge of $v$; set the weight of this task to be the number of messages that have $v$ as their target leaf.  

\begin{figure}
	\centering
	\begin{forest}
	  for tree={%
	    inner sep=0pt,
		circle,
		minimum size=2ex,
	    draw,
		s sep=1ex,
		edge=->
	  },
	  [,name=l1,phantom
	  [,color=black,draw
	  [,color=black,draw
	  	[,color=black,draw
			[,color=black,draw
				[,label=below:$30$,color=red,fill]
			]
		]  
	  ]
	  ]
	  [,color=black,draw
	  [,color=black,draw
	  	[,color=black,draw
			[,color=black,draw
				[,label=below:$10$,color=red,fill]
			]
		]  
	  ]
	  ]
	  [,color=black,draw
	  [,color=nicepurple,fill
	  	[,color=black,draw
			[,color=black,draw
				[,label=below:$3$,color=black,draw]
			]
		]  
	  ]
	  ]
	  [,color=black,draw
	  [,color=black,draw
	  	[,color=black,draw
			[,color=black,draw 
				[,label=below:$5$,color=niceblue,fill ]
				[,label=below:$6$,color=niceblue,fill ]
			]
		]  
	  ]
	  ]
	  [,color=black,draw
	  [,color=black,draw
			[,color=black,draw
				[,color=niceorange,fill
					[,label=below:$6$,color=black,draw ]
					[,label=below:$3$,color=black,draw ]
				]
				[,color=niceorange,fill
					[,label=below:$9$,color=black,draw ]
				]
		  ]
	  ]
	  ]
	  [,color=black,draw
	  [,color=black,draw
		  [,color=black,draw
				[,color=niceyellow,fill
					[,label=below:$9$,color=black,draw ]
				]
				[,color=niceyellow,fill
					[,label=below:$4$,color=black,draw ]
					[,label=below:$5$,color=black,draw ]
				]
		  ]
	  ]
	  ]
	  [,color=black,draw
	  [,color=black,draw
	    [,color=black,draw
	    	[,color=black,draw
				[,label=below:$6$,color=nicegreen,fill ]
	    		[,label=below:$8$,color=nicegreen,fill ]
	    	]
	    ]
	  ]
	  ]
	  [,color=black,draw
	  [,color=black,draw
	  	[,color=purple,fill
	    	[,color=black,draw
	    		[,label=below:$5$,color=black,draw ]
	    		[,label=below:$3$,color=black,draw ]
	    	]
	    	[,color=black,draw
	    		[,label=below:$1$,color=black,draw ]
	    	]
	    ]  
	    [,color=purple,fill
	    	[,color=black,draw
	    		[,label=below:$3$,color=black,draw]
	    		[,label=below:$3$,color=black,draw]
	    	]
	    ]
	  ]
	  ]
	  ]
	\end{forest}
	\caption{This figure shows the precedence constraints of tasks in $\calT(T,M,P,B)$ for the WORMS instance in Figure~\ref{fig:packed_sets}; the coloring of each packed set matches between the figures, except the red packed sets which had a leaf packed parent.  All internal tasks have weight $0$; the leaves are labelled with their weight. If all descendant leaves of a task have weight $0$ the task is omitted.}
	\label{fig:reduction}
\end{figure}
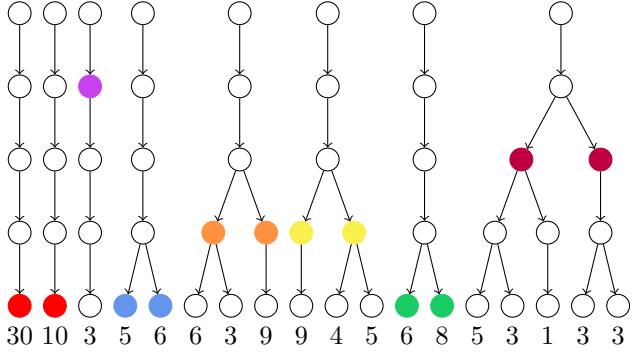

If $v$ is a leaf, set the weight of $j'$ to be the number of messages in $C$.
In Figure~\ref{fig:reduction} we show an example of $\calT(T,M,P,B)$, where $(T,M,P,B)$ is as given in Figure~\ref{fig:packed_sets}.

\begin{lemma}%
\label{lem:scheduling_to_WORMS}
Given a solution $\sigma$ to a $\scheduling$ instance $\calT(T, M, P, B)$, there exists an overfilling schedule $S'$ for $(T, M, P, B)$ with $c(S') = \cost(\sigma)$.
\end{lemma}
\begin{proof}
Each task $j$ in $\calT(T, M, P, B)$ was created for a packed set $C$ and an edge $e\in T$.  When $j$ is processed in $\sigma$, we flush all messages in $C$ across $e$ in $S$.  The precedence constraints guarantee that the flushes are valid.  We have $c(S') = \cost(\sigma)$ because each time a task with nonzero weight completes in $\sigma$, all messages that contribute to the weight of that task reach their leaf in $S'$.
\end{proof}

\begin{lemma}%
\label{lem:WORMS_to_scheduling}
Given an overfilling solution $S$ to a WORMS instance $(T, M, P, B)$, there is a solution $\sigma$ to $\calT(T, M, P, B)$ such that $\cost(\sigma) \leq c_1\cdot c(S)$, with $c_1$ as defined in Lemma~\ref{lem:overfilling_to_valid}.
\end{lemma}
\begin{proof}
	Let $\hat{S}$ be as defined in Section~\ref{sec:converting_overfilling_schedules_to_valid_schedules} using $S$ and the oblivious packed sets of $(T, M, P, B)$.  Note that all flushes in $\hat{S}$ are of messages from a single oblivious packed set.  For $U$ this is by definition; for $L$ this is because all messages whose target leaf is a descendant of a given child of their packed parent are a part of the same oblivious packed set; for $U_r$ and $\hat{S}$ this follows from $L$ and $U$.
Thus, all flushes in $\hat{S}$ correspond to a task in $\calT(T, M, P, B)$.  The lemma follows from Lemma~\ref{lem:overfilling_to_valid}.
\end{proof}

Together, Lemmas~\ref{lem:scheduling_to_WORMS} and~\ref{lem:WORMS_to_scheduling} (with Lemma~\ref{lem:overfilling_to_valid}) show that solving $\calT(T,M, P, B)$ is equivalent, up to constant factors, to solving $(T,M, P, B)$.

\section{Solving $\scheduling$}%
\label{sec:algorithm}

In this section we discuss how to solve any instance $\calT$ of $\scheduling$.  We will use this to give an $O(1)$-approximate solution for $\calT \gets \calT(T, M, P, B)$ to solve a WORMS instance $(T, M, P, B)$.

\subsection{Single-Machine Case}

For the single-processor case (i.e.\ $P = 1$) we are done:  an optimal algorithm for $1|outtree|\sum wC$, Horn's algorithm, was given in~\cite{horn1972single}.  

We give an exposition of Horn's algorithm here.  While the algorithm is already known, we will modify Horn's algorithm to create a parallel algorithm below.

Let $\calT$ be a set of tasks with tree precedence constraints.  
  
For any subtree\footnote{A subtree $\calT'$ of $\calT$ must be a tree (as opposed to a forest---$\calT'$ can have only one root node), and must be contiguous---for any non-root $j\in \calT'$, the parent of $j$ must also be in $\calT'$.} $\calT'$,
its \defn{weight}, denoted $w(\calT')$, is the total weight of all tasks in $\calT'$; the \defn{size} $s(\calT')$ is the number of tasks in $\calT'$.
The \defn{density} of $\calT'$ is $w(\calT')/s(\calT')$.

For any task $j\in \calT$, let $F_j$ be the highest density subtree of $\calT$ rooted at $j$.  We define the \defn{task density} of $j$ to be the density of $F_j$.
Horn gives a algorithm for calculating the task density for all tasks $j$~\cite{horn1972single}.

Horn's algorithm can be defined as follows on an instance $\calT$.  
First, for each task $j\in \calT$, calculate the task density of $j$.  
Then, add all tasks in $\calT$ with no precedence constraints to a priority queue $PQ$, where the priority of a task is its task density.  At each time step, remove the highest-priority task $j'$ from $PQ$ (breaking ties arbitrarily) and process $j'$. Then, add all children of $j'$ to $PQ$.

\begin{lemma}[\cite{horn1972single}]%
\label{lem:hornsoptimal}
Horn's algorithm is an optimal schedule for any instance $\calT$ of $\scheduling$ with $P=1$.  
\end{lemma}

The running time of Horn's algorithm is not given in~\cite{horn1972single}.  However, a priority queue implementation allows both the calculation of all task densities, and the decision of when to run each task, in $O(|\calT|\log |\calT|)$ time.

\subsection{Multiple Machines}
While Horn's algorithm is optimal when $P=1$, in general $\scheduling$ is known to be strongly NP-hard~\cite{Timkovsky03,lenstra1980complexity}.  
The best-known approximation algorithm for $\scheduling$ is the $(1 + \sqrt{2})$-approximation algorithm of~\cite{li2020scheduling} (i.e.\ a $2.415$-approximation).

However, there are downsides to using known algorithms as a black box to solve our problem.
The known approximation algorithms for $\scheduling$ are LP-rounding-based algorithms for $P|prec|\sum wC$.  That is to say, they work for general, not just tree, precedence constraints;  for this reason, they are complicated, and do not give structural insight into designing effective algorithms for the tree case specifically.  
Here, we give a $4$-approximation algorithm that is based on Horn's algorithm: we choose what task to process (i.e.\ what to flush) based on the density of available tasks.

To begin, we define an intermediate algorithm.  
We define an algorithm \defn{Parallel Heaviest Tree First (PHTF)}, which generalizes Horn's algorithm to multiple processors.  As before, we calculate the task density of all tasks in $\calT$, and add all tasks without a precedence constraint to a priority queue $PQ$.  At each time step, we remove the $P$ highest-density tasks from $PQ$ (or all tasks from $PQ$ if fewer than $P$ are available) and process all of them; then, we add the children of any processed tasks to $PQ$.
Our goal is to use PHTF as an intermediate algorithm to give a $4$-approximation algorithm for $\scheduling$.  

As a first step towards that goal, we show that PHTF is optimal under a slightly relaxed cost function.  The idea of this cost function is that it gives a fractional notion of completion time.  

We define the \defn{Horn's trees} for a given instance $\calT$.  
In short, the Horn's trees partition all tasks into dense subtrees processed by Horn's algorithm.
We formally define them recursively: for an instance $\calT'$, find a task $j$ with no parent in $\calT'$; then $F_j$ is a Horn's tree.  Remove all tasks in $F_j$ from $\calT'$ and recurse. This continues until $\calT'$ has no remaining tasks.  

Fix an instance of $\scheduling$  and denote it $\calT$; we denote the set of Horn's trees for $\calT$ as $\calH$.
By definition, every task in $\calT$ is in exactly one Horn's tree $T\in \calH$.

We repeatedly use the following observation about Horn's trees in our proof.

\begin{observation}%
\label{obs:horns_trees_only_densest}
If $T$ and $T'$ are two subtrees of $\calT$, and both $T$ and $T'$ have a common root $r$, then if $T'$ is denser than $T$, $T$ cannot be a Horn's tree.
\end{observation}
\begin{proof}
	If $T$ is a Horn's tree, then $T = F_r$ by definition.  But this is not possible since the density of $F_r$ must be at least the density of $T'$.
\end{proof}

For any schedule $\sigma$ for $\calT$, any Horn's tree $T_i \in \calH$, and any time step $t$, 
let 
$U_i^t(\sigma)$ 
be the subset of tasks in $T_i$ 
that are unfinished in $\sigma$ at time step $t$.  
Observe that we can denote the weighted completion time of $\sigma$ as 
\[
	\cost(\sigma) = \sum_t \sum_{T_i\in \calH} \left(\sum_{j\in U_i^t(\sigma)} w(j)\right).
\]
Then we define $\costf(\sigma)$ for any schedule $\sigma$ as 
\[
	\costf{}(\sigma) = \sum_t \sum_{T_i\in \calH} \left(\frac{|U_i^t(\sigma)|}{s(T_i)}\sum_{j\in T_i} w(j)\right)
\]

In other words, in this notion of cost, the algorithm gets credit for the portion of each tree it completes, where the trees are defined using Horn's algorithm.  
We can equivalently write $\costf(\sigma) = \sum_t \sum_{T_i\in \calH} |U_i^t(\sigma)|w(T_i)/s(T_i)$.

\begin{lemma}%
\label{lem:phtf_optimal_fractional}
PHTF is optimal for $\costf$.
\end{lemma}
\begin{proof}
	Assume by contradiction that the optimal schedule $\sigma$ under $\costf$ satisfies $\costf(\sigma) < \costf(\sigma_H)$.  

	Since $\sigma \neq \sigma_H$, there must be some first time step $t$ where there is a task $j_i$ from the Horn's tree $T_i$ that is one of the $P$ densest trees in the priority queue, but $\sigma$ instead schedules a task from a tree that is strictly less dense than $T_i$.  Let $t_{\ell} > t$ be the time step when $j_i$ is scheduled in $\sigma$.  

	Let $t'$ be the \emph{last} time step before $t_{\ell}$ such that a task $j'$ is processed at $t'$ where $j'$ is from a Horn's tree $T'$ that is strictly less dense than $T_i$. There must be such a $t' \geq t$ because $t' = t$ always satisfies the requirement.

	We obtain a new schedule $\sigma'$ by swapping $j_i$ and $j'$.  To achieve a contradiction, we must prove that $\sigma'$ is a valid schedule (the tasks it performs satisfy the precedence constraints), and that $\costf(\sigma')  < \costf(\sigma)$.

	First, we show that $\sigma'$ is valid.  By definition, any precedence constraint for $j_i$ is satisfied at $t$, and thus is satisfied at $t'$.  Similarly, if any precedence constraint is satisfied for $j'$ at $t'$, it satisfied at $t_\ell > t$.  
	Further, all tasks processed at time steps later than $t_\ell$ or earlier than $t'$ obviously remain valid.
	Now, we must show that all tasks other than $j'$ and $j_i$ executed between $t'$ and $t_\ell$  in $\sigma'$ are valid.
	Consider a task $j''$ processed between $t'$ and $t_\ell$ such that $j''$ is a descendant of $j'$.  Since $j''$ is processed between $t'$ and $t_\ell$, by definition, $j''$ must be from some Horn's tree $T''$ that is more dense than $T_i$, and thus more dense than $T'$.  However, this is not possible by Observation~\ref{obs:horns_trees_only_densest}, because then $T' \cup T''$ would be denser than $T'$ and share a common root.

	Finally, we must show that $\costf(\sigma') < \costf(\sigma)$.  
Let $\delta_i$ and $\delta'$ be the density of $T_i$ and $T'$ respectively.
	We have that $\costf(\sigma) - \costf(\sigma') = (t_\ell - t)(\delta_i - \delta').$  But this is positive since $t_\ell > t$ and $\delta_i > \delta'$.
	Thus, $\sigma'$ has lower cost than $\sigma$, giving a contradiction.
\end{proof}

It's not immediately clear how to compare $\cost{}$ and $\costf{}$.  On the one hand, for a tree $T$ that is a path of tasks of increasing weight, $\costf{}$ always gives a smaller cost than $\cost{}$.  On the other hand, for a non-path $T$, completing a particularly heavy node immediately decreases $\cost{}$ by a large amount, while $\costf{}$ only reduces proportionally to how much of $T$ has been flushed---it is easy to come up with trees (that are not Horn's trees) where $\cost$ is far less than $\costf{}$.
Perhaps surprisingly, we show in the next lemma that despite this tradeoff, $\costf{}$ is always smaller.  The main proof idea is to use the definition of the Horn's trees: 
if completing a portion of $T_i\in \calH$ has smaller cost under $\cost{}$ than under $\costf{}$, then $T_i$ can be split into two trees $T'$ and $T''$ where $T'$ shares a root with $T_i$, but is denser than $T_i$, violating Observation~\ref{obs:horns_trees_only_densest}.

\begin{lemma}%
\label{lem:fractionalcostsmaller}
For any schedule $\sigma$, $\costf(\sigma) \leq \cost(\sigma)$.  
\end{lemma}
\begin{proof}
Consider some Horn's tree $T_i$ and some time step $t$ for the remainder of the proof.
Let $T^f_i$ be the tasks in $T_i$ that are completed by $t$, and $T^u_i$ be the unfinished tasks in $T_i$.  
Then it is sufficient to show that for any $T_i$ and $t$, for all $\sigma$, 
\begin{equation}
	\label{eq:lem_fractional_smaller_rewritten}
	w(T_i) \frac{s(T^u_i)}{s(T_i)} \leq w(T^u_i).
\end{equation}

Since $\sigma$ is a schedule, $T^f_i$ is a subtree of $T_i$ with the same root as $T_i$.  By Observation~\ref{obs:horns_trees_only_densest}, $w(T^f_i)/s(T^f_i) \leq w(T_i)/s(T_i)$.  
We must have $s(T^f_i) + s(T^u_i) = s(T_i)$ and $w(T^f_i) + w(T^u_i) = w(T_i)$, so we can write
\[
	\frac{w(T_i) - w(T^u_i)}{s(T_i) - s(T^u_i)} \leq \frac{w(T_i)}{s(T_i)}.
\]

For any $x,y,a,b > 0$, 
if 
$(x - a)/(y-b) \leq x/y$ then 
$x/y \leq a/b$.  Thus, we have $w(T_i)/s(T_i) \leq w(T^u_i)/s(T^u_i)$.  Rearranging we obtain Equation~\ref{eq:lem_fractional_smaller_rewritten}.
\end{proof}

\paragraph{A $4$-Approximation Algorithm}
We are now ready to define our algorithm, Modified Parallel Heaviest Tree First (MPHTF).  At each time step $t$, let $\calT_t$ be the set of at most $P$ tasks flushed by PHTF.  
For every task $j\in \calT_t$, let $T_j$ be the Horn's tree containing $j$.  Then MPHTF performs a task from $T_j$ (with satisfied precedence constraints) at time step $t$, and performs a task (with satisfied precedence constraints) from $T_j$ at time step $t+1$.  If all tasks in $T_j$ are completed MPHTF does nothing.

\begin{lemma}%
\label{lem:constant_approximation}
If $\sigma_M$ is the solution given by MPHTF on an instance $\calT$ of $\scheduling$, then for any other schedule $\sigma'$ for $\calT$, we have $\cost(\sigma_M) \leq 4\cost(\sigma')$.  
\end{lemma}
\begin{proof}
	Let $\sigma_H$ be the schedule given by PHTF on $\calT$.  By Lemma~\ref{lem:hornsoptimal} we have that $\costf(\sigma_H)\leq \costf(\sigma')$.  Furthermore, from Lemma~\ref{lem:fractionalcostsmaller}, $\costf(\sigma') \leq \cost(\sigma')$.  Therefore, we need only to show $\cost(\sigma_M) \leq 4\costf(\sigma_H)$ to prove the lemma.

At time step $t$, let $\calH_t$ be the set of Horn's trees which have had at most half their edges flushed by PHTF.  Then
\begin{align*}
	\costf{}(\sigma_H) &= \sum_t \sum_{T_i\in \calH} \left(\frac{|U_i^t(\sigma_H)|}{s(T_i)}w(T_i)\right)\\
					   &= \sum_t \sum_{T_i\in \calH_t} \left(\frac{|U_i^t(\sigma_H)|}{s(T_i)}w(T_i)\right)  \\ 
					   & \qquad + \sum_t \sum_{T_i\in \calH\setminus \calH_t} \left(\frac{|U_i^t(\sigma_H)|}{s(T_i)}w(T_i)\right) \\
					   &\geq \sum_t \sum_{T_i\in \calH_t} \left(\frac{1}{2}w(T_i)\right) + \sum_t \sum_{T_i\in \calH\setminus \calH_t} 0
\end{align*}
Let $t_h(i)$ be the first time step when at least half the tasks in Horn's tree $T_i$ completed in PHTF.  Replacing in the above, $\costf(\sigma_H) \geq \sum_{T_i\in \calH_t} t_h(i) w(T_i)/2$.

Since MPHTF flushes any Horn's tree twice when PHTF flushes it once, for any time step $t'$, all tasks in Horn's trees $T_i \in \calH \setminus \calH_t$ are finished by $2t'$ in $\sigma_M$.  Therefore, 
\begin{align*}
	\cost(\sigma_M) & = \sum_t \sum_{T_i \in \calH} \sum_{j\in U^t_i(\sigma)} w(j)\\
	                & = \sum_{T_i \in \calH} \sum_t \sum_{j\in U^t_i(\sigma)} w(j)\\
					& = \sum_{T_i \in \calH} \sum_{t=1}^{2t_h(i)} \sum_{j\in U^t_i(\sigma)} w(j)\\
					&\leq 2\sum_{T_i \in H_{t'}} 2t_h(i) w(T_i).
\end{align*}
Putting the above together, $\cost(\sigma_M)\leq 4\costf(\sigma_H)$.
\end{proof}

We immediately obtain that MPHTF is a $4$-approximation algorithm.

\subsection{Putting it All Together}%
\label{sec:approximating_wots}

We are ready to show how to give, for any WORMS instance $(T,M, P, B)$, a sequence of flushes that $O(1)$-approximates the optimal sequence of flushes.

To begin, we use $(T,M, P, B)$ to calculate $\calT(T,M, P, B)$.  By Lemma~\ref{lem:WORMS_to_scheduling} the best solution to $\calT(T,M, P, B)$ is at most $c_1$ greater than the best solution to $(T,M, P, B)$.  We use MPHTF to $4$-approximate the best solution to $\calT(T,M, P, B)$; call this solution $\sigma_M$.  We then use Lemma~\ref{lem:scheduling_to_WORMS} to convert $\sigma_M$ to an overfilling sequence of flushes $S$, where $S$ has the same cost as $\sigma_M$; at most $4c_1$ of optimal.  Finally, we use Lemma~\ref{lem:overfilling_to_valid} to convert $S$ to $\hat{S}$, a valid sequence of flushes for $(T,M, P, B)$, where $c(\hat{S})$ is at most $4c_1^2$ the optimal cost for $(T,M, P, B)$.  Thus we achieve a schedule $\hat{S}$ that $O(1)$-approximates the optimal schedule for $(T,M, P, B)$.  Each of these operations requires $O(n\log n)$ time, so $\hat{S}$ can be obtained in $O(n\log n)$ total time.

\subsection{NP-hardness}%
\label{sec:np_hardness}

We show that WORMS is NP hard even on a single processor.  This justifies that our results give an approximation of the optimal solution.  The proof has been moved to the appendix for space.

\begin{lemma}%
\label{lem:NP_hard}
WORMS is NP hard even if $P = 1$.
\end{lemma}

\begin{proof}
	We reduce from $3$-partition.  In $3$-partition, we are given a set of $3n'$ integers $I$, where $\sum_{i\in I} i = n'K$ for some $K$, and all $i\in I$ are between $K/4$ and $K/2$.  $3$-partition remains strongly NP-hard even if all $i\in I$ are between $K/4$ and $K/2$ and the integers in $I$ are given in unary.\footnote{This is important because our instance will have $|M|$ polynomial in $\sum_{i\in I} i$.}  The problem is to decide whether there exists a partition of $I$ into $n'$ sets $\calP$, such that each set $\hat{P}\in \calP$ has $\sum_{i\in \hat{P}} i = K$.  (Since all $i$ are between $K/4$ and $K/2$, there must be exactly $3$ items in every set $\hat{P}\in \calP$.)  

	We proceed in two parts.  First, we give an instance of WORMS $(T_1, M_1, 1, B)$ such that $I$ has a $3$-partition if and only if there is a schedule $S_1$ such that: (1) the \emph{maximum} completion time of any message in $(T_1, M_1, 1, B)$ in $S_1$ is at most $4n'$ and (2) the cost of $S_1$ on $(T_1, M_1, 1, B)$ is at most $C_1$.   Then, we give a WORMS instance $(T,M, 1, B)$ that has a schedule of cost at most $C_2$ if and only if $(T_1, M_1, 1, B)$ has such a schedule $S_1$.

First, let us define $T_1$, $M_1$, $B$ and $C_1$.  Let $B=3X+K$ where\footnote{$X$ is a large number added to each integer in $I$---intuitively, adding a large number to each integer means that the integers are proportionately much closer in value; thus, small perturbations in a schedule in our reduction have a minimal impact on cost.} $X=12n'^2K$. To begin, let $T_1$ consist of a root node $r$ and a child $x$ of $r$.  Then, create $3n'$ leaves, all of which are children of $x$: in particular, for each $i\in I$, create a leaf $\ell_i$ in $T_1$, and create $X+i$ messages in $M_1$ that have $\ell_i$ as their target leaf. 
	We call these $X+i$ messages the \defn{representative messages} of $i$.   
	Note that therefore $|M_1| = n'K+3n'X$. Finally, we define: 
\[
	C_1=\sum_{i=1}^{n'} \left(4(i-1)(3X + K) + X(2 + 3 + 4) + 4K\right).
\]

	We show that there is a schedule $S_1$ for $(T_1, M_1, 1, B)$ with maximum completion time $4n'$ and cost at most $C_1$ if and only if $I$ has a $3$-partition.

	If $I$ has a $3$-partition $\calP$, then for each $\hat{P}\in \calP$, we can flush all representative messages of any $i \in \hat{P}$ to $x$; then, for the three $i\in \hat{P}$, we flush the representative messages of each to their leaf.   
	The messages fit in $x$ because $\sum_{i\in \hat{P}} i = K$, so $3X + K$ messages are flushed.
	This requires $4$ flushes; since $|\calP| = n'$ and $P=1$, the maximum completion time of any message is $4n'$.  The sum of completion times of such a schedule is smaller than the one of a virtual schedule that would finish sequentially 0, $X$, $X$, and $X+K$ messages, repeatedly every 4 time steps, which corresponds to the value of $C_1$.  (The $i$th set of messages flushed must wait $4(i-1)$ time steps for the previous sets to finish, after which $X$ messages are completed at time steps $2$, $3$, and $4$, and $K$ messages are completed at time step $4$.)

	On the other hand, assume $(T_1,M_1, 1, B)$ has a schedule $S_1$ that requires at most $4n'$ flushes and a sum of completion times at most $C_1$. 
	Then $S_1$ must have at most $n'$ flushes to $x$ (since each of the $3n'$ leaves must also be flushed to at least once). 
	Call a schedule for $(T_1, M', 1, B)$ \defn{canonical} if it has at most $4n'$ total flushes, and every flush from $r$ to $x$ contains exactly the representative messages of three leaves of $T_1$.  Thus, if $S_1$ is canonical, then the flushes from $r$ to $x$ constitute a $3$-partition.  We now show that $S_1$ is canonical.

	Consider a simpler instance with a set of messages $M' \subset M_1$, where each leaf in $T_1$ has exactly $X$ messages in $M'$ (thus, $M'$ corresponds to an instance where all $i\in I$ have $i = 0$).  We can see that by ignoring messages in $M_1\setminus M'$, ${S}$ is a valid schedule for $(T_1, M', 1, B)$.

	Let's lower bound the cost of any schedule for $(T_1, M', 1, B)$ that requires at most $4n'$ total flushes.  
	Note that since $t$ flushes are required to bring the contents of $3t$ leaves to $x$, no schedule can complete more than $3t$ leaves by time step $4t$ for any $t \geq 0$.  
	Thus, any schedule for $(T_1, M', 1, B)$ has cost
	\begin{align*}
		C_0 &=\sum_{i=1}^{n'} 4(i-1)(3X) + X(2 + 3 + 4) \\
	&= C_1 - \sum_{i=1}^{n'} 4Ki > C_1 - X.
\end{align*}
Clearly, $S_1$ on $(T_1, M', 1, B)$ has cost $C_0$ if and only if it is canonical.  If $S_1$ has cost more than $C_0$ then all messages in some leaf are delayed past their finishing time in a canonical schedule%
\footnote{Recall that since there are $4n'$ flushes each leaf can be flushed to once---if one message is delayed, all messages in that leaf must be delayed as well.} 
in which case the total completion time of $S_1$ on $M'$ (and thus, the total completion time of $S_1$ on $M$) 
is $\geq C_0 + X > C_1$, which contradicts the definition of $S_1$.  Thus, $S_1$ is canonical, and the messages in each flush from $r$ to $x$ constitute a $3$-partition.

Thus, $(T_1, M_1, 1, B)$ has a valid schedule $S_1$ with total completion time $C_1$ and $4n'$ total flushes if and only if $I$ has a $3$-partition.
To finish the reduction, we alter $T_1$ and $M_1$ to obtain a $T$ and $M$ so that $(T, M, 1, B)$ has a valid schedule with total completion time $C_2$ if and only if $I$ has a $3$-partition.

To create $T$ and $M$, add another $8n'|M_1|+C_1$ leaves to $T_1$ to obtain $T$; these leaves all have a unique parent (thus we add many paths of length $2$ to $T_1$).  For each leaf $\ell$ created, add a message with $\ell$ as its target leaf.  Let $M_2$ be all the messages created in this way, and $M = M_1\cup M_2$.  See Figure~\ref{fig:nphard}.

\begin{figure}
\begin{center}
\begin{forest}
  for tree={%
    inner sep=0pt,
    minimum size=10pt,
    circle,
    draw
  },
  [,label=left:$r$
  	[,label=left:$x$,s sep=6ex
		[,label=below:$X + i_1$]
		[,label=below:$X + i_2$]
		[$\ldots$,no edge,draw=none]
		[,label=below:$X + i_{3n'}$]
	]
	[,no edge,phantom [,no edge,phantom]]
	[,edge=dashed [,label=below:$1$,edge=dashed] ]
	[,edge=dashed [,label=below:$1$,edge=dashed] ]
	[,phantom,no edge [,label=above:$\ldots$,no edge,draw=none] ]
	[,edge=dashed [,label=below:$1$,edge=dashed] ]
  ]
\end{forest}
\end{center}%
\caption{A diagram of $T$ in Lemma~\ref{lem:NP_hard}.  Edges in $T_2$ are given as dotted lines; edges in $T_1$ are solid.  The number of all messages in $M$ with a given target leaf is given below each leaf.}%
\label{fig:nphard}%
\end{figure}
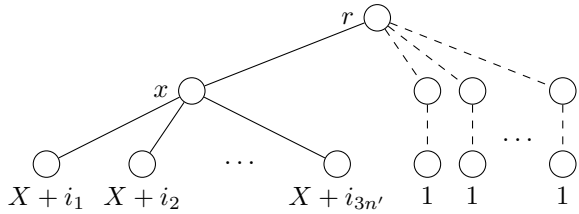

We define 
\[
	C_2 = C_1 + 4n'|M_2| + \sum_{i = 1}^{|M_2|} 2i.
\]

If all messages in $(T_1, M_1, 1, B)$ can be finished by $4n'$ and with a sum of completion time at most $C_1$ by some schedule $S_1$, then we can achieve a schedule that finishes with cost $C_2$: first, schedule all messages in $M_1$ using $S_1$, then flush all messages in $M_2$ one at a time.  The weighted completion time of all messages in $M_1$ is $C_1$; the messages in $M_2$ must wait at most $4n'$ time steps for all messages in $M_1$ to finish, after which they finish one at a time, each requiring $2$ flushes.

Now, assume that all messages in $(T_1, M_1, 1, B)$ can finish by time $4n'+1$ at the earliest.  Let $S$ be the optimal schedule for $(T, M, 1, B)$.  By an exchange argument, all messages in $M_1$ finish before all messages in $M_2$ in $S$.  Thus, all messages in $M_2$ must wait for all messages in $M_1$ to complete, after which they can be finished one at a time.  Thus, the cost for messages in $M_2$ in $S$ is at least 
\[
	(4n'+1)|M_2| + \sum_{i = 1}^{|M_2|} 2i > C_2
\]
since $|M_2| > 4n'|M_1| + C_1$.

Now, assume that all messages in $(T_1, M_1, 1, B)$ cannot finish with a sum of completion time at most $C_1$. We have seen that the maximum completion time of these messages is at most $4n'$. Let $S$ be the optimal schedule for $(T, M, 1, B)$. Again, by an exchange argument, all messages in $M_1$ finish before all messages in $M_2$ in $S$. Thus, the cost for messages in $M_2$ in $S$ is at least $C_2+1$.  

Thus, $(T, M, 1, B)$ has a schedule with cost $C_2$ if and only if $I$ has a $3$-partition.

\end{proof}

\section{Conclusion}%
\label{sec:conclusion}
This paper formalizes WORMS, a new model of root-to-leaf flushing in \wods.
WORMS captures a new notion of latency in write-optimized trees: 
some types of operations complete only after they have flushed through their entire root-to-leaf path in the tree.
The goal is to complete these queries in the fewest cache misses possible,
thus minimizing the query latency.

We give a $O(1)$-approximation algorithm for WORMS via a reduction to a closely-related scheduling problem, and we show that solving WORMS exactly is NP-hard.

The natural next step is to generalize to the much more challenging online and non-clairvoyant setting: messages arrive online, and the path each message takes through the tree is not known to the scheduler.  
Our hope is that the structure of the algorithms presented here can help inform strategies for this more general problem.

\section{Acknowledgments}

This work was supported in part by NSF grants CCF 210381 and CNS 1938709.

\bibliographystyle{plain}
\bibliography{upserts}

\end{document}